\documentclass[runningheads, envcountsame, a4paper]{llncs}

\usepackage[utf8]{inputenc}
\usepackage{bbold}
\usepackage{amsmath}
\usepackage{amssymb}
\usepackage{color}
\usepackage{microtype}
\newcommand\LQ[1]{#1\backslash}
\newcommand\RQ[1]{#1 /}

\newcommand{\ba}{\begin{array}}
\newcommand{\ea}{\end{array}}

\newcommand{\True}{\Rempty}
\newcommand{\False}{\Rnull}

\newcommand\Power{\wp}
\newcommand\nat{\mathbb{N}}
\newcommand\bool{\mathbb{B}}

\newcommand\Rnull{\mathbf0}
\newcommand\Rempty{\mathbf1}

\newcommand\Lang{\mathcal{L}}
\renewcommand\L\Lang

\newcommand\calF{\mathcal{F}}
\newcommand\calR{\mathcal{R}}
\newcommand\calI{\mathcal{I}}
\newcommand\calJ{\mathcal{J}}
\newcommand\hcalI{\hat\calI}
\newcommand\hcalJ{\hat\calJ}
\newcommand\Arity{\#}

\newcommand\IsExtensionOf{\succ}
\newcommand\IsExtensionOrEqual{\succeq}

\newcommand\Upward[1]{{\uparrow}{#1}}
\newcommand\Downward[1]{{\downarrow}{#1}}

\newcommand\NF{\textsf{nf}}

\title{Derivatives for Enhanced Regular Expressions}
\author{Peter Thiemann}
\institute{University of Freiburg}

\begin{document}

\maketitle
\setcounter{footnote}{0}
\begin{abstract}
  Regular languages are closed under a wealth of formal language
  operators. Incorporating such operators in regular expressions leads
  to concise language specifications, but the transformation of
  such enhanced regular expressions to finite automata becomes more
  involved.

  We present an approach that enables the direct construction of
  finite automata from regular expressions enhanced with further
  operators that preserve regularity. Our construction is based on an extension of the theory of
  derivatives for regular expressions. To retain the standard results
  about derivatives, we develop a derivability criterion for the compatibility of
  the extra operators with derivatives.

  Some derivable operators do not preserve regularity. Derivatives
  provide a decision procedure for the word problem of
  regular expressions enhanced with such operators. 
\end{abstract}
\section{Introduction}
\label{sec:introduction}

Brzozowski derivatives \cite{321249} and Antimirov's partial derivatives \cite{Antimirov96Partial} are
well-known tools to transform  regular expressions to automata and to define algorithms for
equivalence and containment on them \cite{AntimirovIneq,Grabmayer:2005:UPC:2156157.2156171}.
Brzozowski's automaton construction relies on the finiteness of the
set of iterated derivatives when considered up to similarity
(commutativity, associativity, and idempotence for union).  
Derivatives had quite some impact on the study of
algorithms for regular languages on finite words and trees \cite{DBLP:conf/rta/RouV03,CaronChamparnaudMignot2011}.

While derivative-based algorithms have been deprecated for performance
reasons \cite{DBLP:conf/wia/Watson96}, there has been renewed interest in the study of
derivatives and partial derivatives. On the practical side, Owens and coworkers \cite{re-derivs}
report a functional implementation that revives many features. Might
and coworkers \cite{DBLP:conf/icfp/MightDS11} implement parsing for 
context-free languages using derivatives.

A common theme on the theory side is the study of derivative structures for enhancements of regular expressions.
While Brzozowski's original work covered extended regular
expressions, partial derivatives were originally limited to simple expressions without intersection and
complement. It is a significant effort to define partial derivatives for extended regular expressions
\cite{CaronChamparnaudMignot2011}.  

Derivatives have also  been used to study various shuffle
operators for applications in modeling concurrent programs
\cite{DBLP:conf/lata/SulzmannT15}. Later extensions consider forkable expressions with a
new operator that abstracts process creation \cite{SulzmannThiemann2016-lata}. 

Caron and coworkers \cite{DBLP:conf/wia/CaronCM12} study derivatives for multi-tilde-bar
expressions. The tilde (bar) operator adds (removes) $\varepsilon$ from a language. Multi-tilde-bar
applies to a list of languages and (roughly) defines a selective concatenation operation that can be
configured to include or exclude certain languages of the list.

Champarnaud and coworkers \cite{DBLP:conf/lata/ChamparnaudJM12} consider derivatives of approximate regular
expressions (ARE). AREs extend regular expressions with a family of unary operators $\mathbb{F}_k$, for
$k\in\nat$, which enhance their argument language $L$ with all words $u$ such that $d (u, w) \le k$,
for some word $w \in L$. Here, $d$ is a suitable distance function, for example, Hamming distance or
Levenshtein distance.

Traytel and Nipkow \cite{DBLP:journals/jfp/TraytelN15} obtain
decision procedures for MSO using a suitably defined derivative
operation on regular expressions with a projection operation.

The general framework of Caron and coworkers
\cite{DBLP:journals/ita/CaronCM14} generalizes the syntactic structure
underlying a derivative construction to a \emph{support}. A support
generalizes expressions (for constructing Brzozowski derivatives),
sets of expressions (for Antimirov's partial derivatives), and sets of
clausal forms over sets of regular expressions, and thus yields an
encompassing framework in which different kinds of derivative
constructions can be formalized and compared. The authors give a
sufficient criterion for a support to generate a finite number of
iterated derivatives from a given expression along with automata
constructions for deterministic, nondeterministic, and alternating
finite automata. Their work applies to
extended regular expressions with arbitrary boolean functions.

Champernaud and coworkers~\cite{ChamparnaudMignotNicart2015} propose
constrained regular expressions with a notion of comprehension
(filtering by a predicate) and matching. The resulting languages are
in general not regular, but the authors propose expression derivation
for the membership test and study its decidability. 

\subsection*{Contributions}
\label{sec:overview}

In this work, we identify a pattern in the definition of (standard) derivatives for enhancements of
regular expressions that go beyond boolean functions. Concretely, we
consider regular expressions enhanced with further operators on
languages (e.g., shuffle, homomorphism, approximation).
Then we propose \emph{left derivability} and
\emph{$\varepsilon$-testability} as a sufficient condition for the
set of operators such that a syntactic derivative operation is
definable for enhanced expressions. This condition gives rise to a
decision procedure for the membership test for enhanced expressions
via expression derivation.

A refinement, \emph{linear left derivability}, is a sufficient
condition to guarantee finiteness of the set of
dissimilar derivatives of an enhanced expression. The finiteness
condition enables the direct construction of a deterministic finite
automaton. We show that every linear left derivable operator can be defined by a
rational finite state transducer and thus preserves
regularity.

Proofs and further examples may be found in the appendix.

\section{Preliminaries}
\label{sec:preliminaries}

We write $\nat$ for the set of natural numbers, $\bool = \{ \False, \True \}$ for the set of
booleans, and $X \uplus Y$ for the disjoint union of sets $X$ and
$Y$. We sometimes write $(\overline{E_k}^{k=1,\dots,n})$ for the tuple
$(E_1, \dots, E_n)$ where $E_k$ is some entity depending on $k$.

An alphabet $\Sigma$ is a finite set of symbols. The set $\Sigma^*$ denotes the set of finite words
over $\Sigma$, $\varepsilon\in\Sigma^*$ stands for the empty word, and $\Sigma^+ = \Sigma^*
\setminus \{\varepsilon\}$. For
$u,v, w\in\Sigma^*$, we write $|u|\in\nat$ for the length of $u$, $u\cdot v $ (or just $uv$) for the
concatenation of words, and $w \IsExtensionOf v$ if $v$ is a proper suffix of $w$, that is, $\exists
u\in \Sigma^+$ such that $w=u\cdot v$.

Given languages $U,V,W \subseteq \Sigma^*$, 
concatenation extends to languages as usual: $U
\cdot V = \{ u \cdot v \mid u\in U, v \in V\}$.
The Kleene closure is defined as the smallest set $U^* \subseteq \Sigma^*$ such that $U^* = \{\varepsilon\} \cup U
\cdot U^*$.
We write the \emph{left quotient} as $\LQ{U}W
= \{ v \mid v\in \Sigma^*, \exists u\in U: uv \in W \}$ and the \emph{right quotient} as $\RQ{W}{U}
= \{ v \mid v\in \Sigma^*, \exists u\in U: vu \in W \}$. For a singleton language $U = \{u\}$, we write $\LQ{u}W$ ($\RQ{W}{u}$)
for the left (right) quotient.

A \emph{ranked alphabet} $\calF$ is a finite set of \emph{operator symbols} with a function
$\Arity : \calF \to \nat$ that determines the \emph{arity} of each symbol. We write $\calF^{(n)} = \{
F\in\calF \mid \Arity (F) = n \}$ for the symbols of arity $n$. The set \emph{$T_\calF (X)$
  of $\calF$-terms over a set $X$} is defined inductively. If $x \in X$, then
$x \in T_\calF (X)$. If $n\in\nat$, $F \in \calF^{(n)}$, and
$t_1, \dots, t_n \in T_\calF (X)$, then $F (t_1, \dots, t_n) \in
T_\calF (X)$.

An \emph{$\calF$-algebra} consists of a carrier set $M$ and an interpretation function $\calI :(n:\nat)
\to \calF^{(n)} \to M^n \to M$. Given a function $\calI_0 : X \to M$, the \emph{term interpretation} $\hcalI({t})$,
for $t \in T_\calF (X)$, is defined inductively as follows. If $x \in X$, then $\hcalI({x}) =
\calI_0(x)$. If $F\in\calF^{(n)}$ and $t_1,\dots, t_n \in T_\calF (X)$, then $\hcalI (F (t_1, \dots,
t_n)) = \calI (n) (F) (\hcalI (t_1), \dots, \hcalI (t_n))$.  We often
write $T_\calF$ in place of $T_\calF (\emptyset)$. 

To avoid notational clutter, we fix an arbitrary alphabet $\Sigma$.

\begin{definition}\label{def:regular-expression}
  The \emph{regular alphabet} is defined by $\calR = \Sigma \uplus \{ \Rnull, \Rempty, {\cdot}, {+}, {*} \}$ with
  arities $\Arity (x) =0$, for $x\in\Sigma$, $ \Arity (\Rempty) = \Arity (\Rnull) = 0$, $\Arity
  (*) = 1$, and $\Arity ({\cdot}) = \Arity ({+}) 
  = 2$.

  Similarity is defined as the smallest equivalence relation
  ${\equiv} \subseteq T_\calR \times T_\calR$ that enforces left and right unit, idempotence,
  commutativity, and associativity for the $+$ operator. For all $r, s, t \in T_\calR$, the relation
  $\equiv$ contains the pairs:
  \begin{align*}
    r + \Rnull &\equiv r  & \Rnull + s & \equiv s &
                                                    r + r & \equiv r &
               r+s &\equiv s+r &  (r+s)+t &\equiv r+ (s+t)
  \end{align*}

  The set $R$ of regular expressions over $\Sigma$ is defined as the
  quotient term algebra $R = T_{\calR}/ (\equiv)$. 

  The language of $r\in R$ is defined by  $\Lang (r) = \hcalI (r)$, that
  is, the interpretation of the term in the $\calR$-algebra with carrier set $\Power(\Sigma^*)$ and
  interpretation function
  \begin{displaymath}
    \begin{array}[t]{ll}
    \calI (0) (\Rnull) &= \{\} \\
    \calI (0) (\Rempty) &= \{\varepsilon\} \\
    \calI (0) (x) &= \{ x \} 
    \end{array}
    \qquad
    \begin{array}[t]{ll}
    \calI (1) (*) &= U \mapsto U^* \\
    \calI (2) ({\cdot}) &= (U,V) \mapsto U\cdot V \\
    \calI (2) (+) &= (U,V) \mapsto U \cup V
    \end{array}
  \end{displaymath}
\end{definition}
The interpretation function $\calI$ is compatible with the definition
of $R$ as a quotient term algebra because the interpretation of $+$
maps equivalent expressions to the same language. We usually work with
a unique representative for each equivalence class computed by a
function $\NF$ (see \cite{Grabmayer:2005:UPC:2156157.2156171}).
We use parenthesized infix notation for the binary operators $\cdot$ and $+$ and postfix
superscript for the unary $*$. We adopt the convention that $\cdot$
binds stronger than $+$ to omit parentheses. The overloading of $\Rnull$ and $\Rempty$ as regular expressions and
boolean values is deliberate.
\begin{definition}
  The operations $\odot, \oplus : R \times R \to R$ are \emph{smart
    concatenation} and \emph{union} constructors for regular
  expressions. Operator $\odot$ binds stronger than $\oplus$.
  \begin{displaymath}
    r \odot s =
    \begin{cases}
      \Rnull & r=\Rnull \vee s=\Rnull \\
      r & s=\Rempty \\
      s & r=\Rempty \\
      (r\cdot s) & \text{otherwise}
    \end{cases}
    \qquad
    r \oplus s = \NF (r + s)
  \end{displaymath}
\end{definition}
\begin{lemma}
  For all $r$, $s$: 
  $\Lang (r \odot s) = \Lang (r\cdot s)$; $\Lang (r \oplus s) = \Lang (r+s)$.
\end{lemma}
\begin{definition}
  A regular expression $r$ is \emph{nullable} if $\varepsilon \in 
  \Lang(r)$.
  The function $N :R \to \{\Rnull, \Rempty\}$ detects nullable expressions:
  $N(\Rempty) = \Rempty$.
  $N(\Rnull) = \Rnull$.
  $N(x) = \Rempty$.
  $N(r\cdot s) = N(r) \odot N(s)$.
  $N(r+s) = N(r) \oplus N(s)$.
  $N(r^*) = \Rempty$.
\end{definition}
\begin{lemma}\label{lemma:nullable}
  For all $r \in R$. $N(r) = \Rempty$ iff $\varepsilon \in \Lang(r)$.
\end{lemma}

\begin{definition}\label{definition:brzozowski-derivative}
  The \emph{Brzozowski derivative} \cite{321249} is a function $D : \Sigma \times
  T_\calR \to T_\calR$ defined inductively for $a\ne b\in \Sigma$ and
  $r, s\in T_\calR$.
  \begin{align*}
    D (a, \Rnull) &= \Rnull      &     D (a, r+s) &= D (a,r) \oplus D (a, s) \\
    D (a, \Rempty) &= \Rnull     &     D (a, r\cdot s) &= D (a,r) \odot s \oplus N (r)  \odot D (a,s)\\
    D (a, a) &= \Rempty          &    D (a,r^*) &= D (a,r) \odot r^*
    \\
    D (a, b) &= \Rnull
  \end{align*}
  It extends to a function on words and languages $D : \Sigma^* \times T_\calR \to
  T_\calR$ and $D : \Power (\Sigma^*) \times T_\calR \to \Power
  (T_\calR)$ as usual  ($a\in\Sigma$, $w\in\Sigma^*$, $U\subseteq \Sigma^*$): 
  \begin{align*}
    D (\varepsilon, r) &= r
    & D (a\cdot w, r) &= D (w, D (a, r))
    & D (U, r) &= \{ D (w, r) \mid w \in U \}
  \end{align*}
\end{definition}
\begin{theorem}[\cite{321249}]\label{theorem:word-derivative}
  For all $w\in\Sigma^*$, $r\in T_\calR$, 
  $\Lang (D (w, r)) = w \setminus \Lang (r)$.
\end{theorem}
\begin{theorem}[\cite{321249}]
  For all $r \in T_\calR$,
  $\Lang(r) =\Lang\Big( N(r) + \sum_{a\in\Sigma} D (a, r) \Big)$.
\end{theorem}

\begin{definition}
  A (nondeterministic) finite automaton (NFA) is a tuple $\mathcal{A} = (Q, \Sigma, \delta, q_0, F)$ where $Q$ is a finite
  set of states, $\Sigma$ an alphabet, $\delta : Q \times \Sigma \to \Power (Q)$ the transition
  function, $q_0 \in Q$ the initial state, and $F \subseteq Q$ the set of final states.

  Let $n\in\nat$.
  A \emph{run} of $\mathcal{A}$ on $w = a_0\dots a_{n-1}\in \Sigma^*$ is a sequence $q_0\dots q_n \in
  Q^*$ such that, for all $0\le i <n$,
  $q_{i+1} \in \delta (q_i, a_i)$. The run is \emph{accepting} if $q_n \in F$.
  The language
  $\Lang (\mathcal{A}) = \{ w \in \Sigma^* \mid \exists
  \textrm{ accepting run of $\mathcal{A}$ on $w$} \}$ is recognized by $\mathcal{A}$.

  The automaton $\mathcal{A}$ is total deterministic if $|\delta (q,a)| =1$, for all $q\in Q$, $a\in\Sigma$.
\end{definition}

\section{Enhanced Derivatives}
\label{sec:abstract-derivatives}

An operation on languages takes one or more languages as arguments and yields another language. 
In this section, we enhance the syntax and semantics of regular expressions with extra operations and consider
conditions for the existence of a syntactic derivative for such enhanced expressions. Many examples
can be drawn from the closure properties of regular languages.
\begin{definition}
  A function $f : (\Sigma^*)^n \to \Sigma^*$ is \emph{regularity-preserving} if for all regular
  languages $R_1, \dots, R_n$ the image $f (R_1, \dots, R_n)$ is a regular language.
\end{definition}
\begin{example}\label{example:regularity-preserving-operators}
  We give a range of examples for operators on languages. All
  operators,  except shuffle closure,  are
  regularity-preserving. Proofs may be found in textbooks on formal
  languages unless otherwise indicated. We let $U, V, L \subseteq \Sigma^*$
  range over regular languages; $a, b \in \Sigma$ range over symbols.
  \begin{enumerate}
  \item The intersection $U \cap V$ and the complement $\neg U$  of
    regular languages are regular.
  \item The shuffle of two regular languages is defined by
    $U\|V = \bigcup \{u \| v \mid u\in U, v \in V\}$ where $\varepsilon \| v = \{v\}$, $u \|
    \varepsilon = \{u\}$, and $au \| bv = \{a\} \cdot (u \| bv) \cup
    \{b\} \cdot (au \| v)$, is regular.\\
    The shuffle closure operation
    $L^\| = \{\varepsilon\} \cup L \cup (L \| L) \cup (L \| L \| L)
    \cup \dots$ does \textbf{not} preserve regularity.
  \item The inverse homomorphism, i.e., $h^{-1} (U) = \{ w \in
    \Sigma^* \mid h (w) \in U \}$ is regular for a 
    function $h : \Sigma \to \Sigma^*$ that is extended  
    homomorphically to a function $\Sigma^* \to \Sigma^*$ (for simplicity, we do not consider
    homomorphisms between different alphabets, which can be simulated by using the disjoint union of
    the alphabets).\\
    The non-erasing homomorphism $h (L) = \{ h (w) \mid w \in L \}$ is
    regular where $h: \Sigma \to \Sigma^+$.
  \item  The language of every $k$-th symbol starting from position
    $i$ from words in a regular language $L$ is regular:
    for $k>0$ and $0<i\le k$ \\
    $f_{i,k} (L) = \{ a_ia_{i+k}a_{i+2k}\cdots a_{i+k\lfloor
      (n-i)/k\rfloor} \mid n\in\nat,
    a_1\dots a_n \in L \}$.
  \item The left quotient $\LQ{}{}$ and the right quotient $\RQ{}{}$
    of regular languages are regular.
  \item Functions $\textit{suffixes} (L) = \Sigma^* \setminus L$
    and  $\textit{prefixes} (L) = L / \Sigma^*$ preserve regularity.
  \item The function $\textit{reverse} (L) = \{ a_n \cdots a_1 \mid n,i\in\nat, 1\le i
    \le n, a_i \in \Sigma,  a_1 \dots a_n \in L \}$ preserves regularity.
  \item For each $k\in\nat$, the function $\mathbb{H}_k (L) = \{ v \mid v\in\Sigma^*,
    \exists u \in L. d (u, v) \le k \}$ is regularity preserving where the
    \emph{Hamming distance} of words $a_1\cdots
    a_n$ and $b_1\cdots b_m$ is defined by $h = d (a_1\cdots
    a_n,b_1\cdots b_m)$. If $m = n$, then $h = |\{ i \mid 1\le i
    \le n, a_i \ne b_i\}|$. Otherwise $h = \infty$.
    
    Analogously, $\mathbb{L}_k (L)$ is a regularity preserving
    approximation that uses the Levenshtein distance (see \cite{DBLP:conf/lata/ChamparnaudJM12}).
  \item The tilde and bar operators defined by $\tilde{L} = L \cup
    \{\varepsilon\}$ and $\bar{L} = L \setminus \{\varepsilon\}$
    preserve regularity (they are the
    primitive building blocks of multi-tilde-bar expressions
    \cite{DBLP:conf/wia/CaronCM12}, which we do not consider
    to save space).
  \end{enumerate}
\end{example}
The notion of a nullable expression is an important ingredient in the definition of the derivative
(Definition~\ref{definition:brzozowski-derivative}).  Nullability can be computed by induction on a regular
expression because each regular operator corresponds to a boolean function on the nullability of the
operator's arguments. The following definition imposes exactly this condition on the extra
operators in regular expressions. 
\begin{definition}
  A function $f : (\Sigma^*)^n \to \Sigma^*$ is \emph{$\varepsilon$-testable}, if there is a
  boolean function $B_f : \bool^n \to \bool$ such that $\varepsilon \in f (L_1, \dots, L_n)$ iff
  $B_f ((\varepsilon \in L_1), \dots, (\varepsilon \in L_n))$.
\end{definition}
\begin{example}
  Some of the functions from Example~\ref{example:regularity-preserving-operators} are
  $\varepsilon$-testable.
  \begin{enumerate}
  \item intersection, complement: $B_\cap = \wedge$, $B_\neg = \neg$;
  \item shuffle: $B_\| = \wedge$;
    the shuffle closure operation
    $L^\|$ is $\varepsilon$-testable
    using $B^\| (b) = \True$;
  \item inverse homomorphism: $B_{h^{-1}} (b) = b$, for $b\in\bool$;
    homomorphism $h$: if $h$ is non-erasing, then $B_h (b) = b$; erasing
    homomorphism is not $\varepsilon$-testable: consider $L_1 = \{ a
    \}$, $L_2 = \{ b \}$, and an erasing homomorphism $h$
    defined by $h(a)=\varepsilon$ and $h(b)=b$. Thus, $h (L_1) =
    \{\varepsilon\}$ and $h (L_2) = \{b\}$. If there was a boolean
    function $f_h$ to vouch for $\varepsilon$-testability of $h$, then $L_1$
    shows that $f_h (\False) = \True$ and $L_2$ yields $f_h (\False) =
    \False$, a contradiction.
  \item $k$-th letter extraction:
    $\varepsilon \in f_{i,k} (L)$ if $\exists w \in L$ such that $|w|<i$, so $f_{i,k}$ is not
    $\varepsilon$-testable. To see this let $i=k=2$, $L_1 = \{a\}$, and
    $L_2 = \{aa\}$ and assume that $B_f$ is the boolean function
    required for $\varepsilon$-testability.  Now $f_{2,2} (L_1) =
    \{\varepsilon\}$ and $f_{2,2} (L_2) = \{a\}$, so that $B_f
    (\False) = \True$ (by $L_1$) and $B_f (\False) = \False$ (by
    $L_2$), a contradiction.
  \item The left quotient is not $\varepsilon$-testable because
    $\varepsilon\in \LQ{U}{W}$ iff $U\cap W \ne \emptyset$:
    consider $U = \Sigma^*$ with $a\in\Sigma$, $W_1 = \emptyset$, and $W_2 = \{a\}$ so that
    $\LQ{U}{W_1} = \emptyset$ and $\LQ{U}{W_2} = \{\varepsilon, a\}$. A binary boolean
    function $B_\setminus$ for $\varepsilon$-testability would have to
    satisfy  $B_\setminus (\False, \False) = \False$ (for $W_1$) and
    $B_\setminus (\False, \False) = \True$ (for $W_2$), a contradiction.
    The same reasoning applies, mutatis mutandis, to the right quotient.
  \item The \textit{suffixes} function is not $\varepsilon$-testable
    by the proof for the left quotient. The proof for
    \textit{prefixes} is analogous to the one for the right quotient.
  \item The \textit{reverse} function is $\varepsilon$-testable: $B_{\textit{reverse}} (b) = b$.
  \item The approximation for Hamming distance is
    $\varepsilon$-testable by $B_{\mathbb{H}_k} (b) = b$ . The approximation for
    Levenshtein distance $\mathbb{L}_k$ is not $\varepsilon$-testable
    for $k>0$. The argument here is similar as for erasing
    homomorphism because a word at distance $k$ from a given word $w$ may be up to $k$
    symbols shorter than $w$.
  \item The tilde and bar operators are obviously
    $\varepsilon$-testable with the constants $\Rempty$ and $\Rnull$, respectively.
  \end{enumerate}
\end{example}
\begin{definition}[Enhanced regular expression]\label{def:regular-augmentation}
  Let $\calF \supseteq \calR$ be a ranked alphabet, an \emph{enhanced
    regular alphabet}. Let further  $\calJ$ be an interpretation
  function for $\calF$ on the carrier $\Power (\Sigma^*)$ extending
  the regular interpretation $\calI$ from 
  Definition~\ref{def:regular-expression}.
  The set of \emph{$\calF$-regular expressions over a set $X$} is the set of terms
  $T_{\calF} (X)$.
  For  $t\in T_{\calF} (X)$ we define its language $\L (t) = \hcalJ
  (t)$.
  The resulting $\calF$-algebra $(\Power(\Sigma^*), \calJ)$ is
  a \emph{regular enhancement} if every symbol
  $F \in\calF^{(n)}$ is interpreted by a regularity-preserving function $\calJ (n) (F)$.
\end{definition}
\begin{example}
  To extend regular expressions with a shuffle operator, consider
  $\calF^\| = \calR \cup \{ \| \}$ with $\Arity{\|} = 2$.

  To extend expressions with $k$th-letter extraction, we consider
  $\calF^{x-k} = \calR \cup \{ f_{i,k} \mid 0<i\le k \}$ with
  $\Arity{f_{i,k}} = 1$.
\end{example}
\begin{lemma}\label{lemma:extended-nullability}
  If $\calJ (F)$ is $\varepsilon$-testable, for each $F\in\calF$, then the nullability function $N$
  can be extended to $\calF$.
\end{lemma}

To obtain syntactic derivability for an enhanced regular expression, it must be possible to express
the derivative of an operator in terms of a regular expression that applies the derivative to the
arguments of the operator.
We first define a suitable property semantically as an algebraic property of a regular enhancement.
\begin{definition}\label{def:left-closed}
  Let $\calF$ be an enhanced regular alphabet and $\calJ$ an extension
  of the regular interpretation $\calI$.
  The $\calF$-algebra $(\Power (\Sigma^*), \calJ)$ is \emph{left derivable} if,
  for each $F\in \calF^{(k)}$  and $a\in\Sigma$, there exists
  a finite subset $X \subset \{ x_{v,j} \mid v \in \Sigma^*, 1\le j\le k\}$ and
  an $\calF$-regular expression $r \in T_{\calF} (X)$
  such that, for all  $L_1, \dots, L_k \subseteq \Sigma^*$
  the left
  quotient 
  $\LQ{a} (\calJ(F) (L_1, \dots, L_k))$ can be expressed as $\hcalJ (r)$ using the interpretation $\calJ_0 (x_{v,j}) = \LQ{v}{L_j}$.
\end{definition}

\begin{example}
  We revisit the previous examples of functions on languages and examine them for being left derivable.
  \begin{enumerate}
  \item Intersection is left derivable: $\LQ{a} (L_1 \cap L_2) = \hcalJ (x_{a,1} \cap x_{a,2}) = \LQ{a}L_1 \cap
    \LQ{a}L_2$. For negation $\neg$, the pattern is the same.
  \item Shuffle is left derivable:\\
    $\LQ{a} (L_1 \| L_2) = (\LQ{a}L_1) \| L_2 \cup L_1 \| (\LQ{a}L_2) = \hcalJ
    (x_{a,1} \| x_{\epsilon,2} + x_{\epsilon,1} \| x_{a,2})$;\\
    shuffle closure is also left derivable:
    \begin{displaymath}
      \LQ{a}{L^\|}
      =(\LQ{a}{L}) \| L^\|
      = \hcalJ(x_{a,1} \| x_{\epsilon,2} + x_{\epsilon,1}^\|)
    \end{displaymath}
  \item Inverse homomorphism is left derivable:\\
    $\LQ{a} (h^{-1} (L)) = h^{-1} (\LQ{h (a)} L) = \hcalJ (h^{-1} (x_{h
      (a), 1}))$. \\
    Non-erasing homomorphism is left derivable:\\
    $\LQ{a} (h (L)) = \bigcup_{b\in\Sigma, h (b) = av} v \cdot h (\LQ{b}{L}) = \hcalJ
    (\sum_{b\in\Sigma, h (b) = av} v\cdot h(x_{b,1})) $.
  \item For $k>1$, the set  $ \{ f_{i,k} \mid 0<i\le k\}$ is left derivable.\\
    $\LQ{a} (f_{i,k} (L)) = f_k (\bigcup_{|w|=i-1}\LQ{wa}L) = \hcalJ (\sum_{|w|=i-1} f_k (
    x_{wa,1}))$.
  \item The left and right quotients are left derivable.\\
    $\LQ{a} (\LQ{L_1} L_2)  = \LQ{(L_1\cdot a)} L_2 = \hcalJ (\LQ{(x_{\varepsilon,1}\cdot a)}
    x_{\varepsilon,2})$.\\
    $\LQ{a} (\RQ{L_1}{L_2}) = \RQ{(\LQ{a}L_1)}{L_2} = \hcalJ (\RQ{x_{a,1}}{x_{\varepsilon, 2}})$.
  \item The function \textit{suffixes} is not left derivable because
    $\LQ{a}{\mathit{suffixes} (L)} = \{ w \mid \exists u. u a w \in L
    \} = \LQ{(\Sigma^*\cdot a)}{L}$ cannot be finitely expressed using just derivatives, the
    \textit{suffixes} function, and the regular operators.

    To see this, consider the family of languages $L_n = w_n^*$ where
    $w_n =(abab^2\cdot ab^n)^*$, for all $n\in\nat$, and find that
    \begin{align*}
      L_n' = \LQ{a}{\mathit{suffixes} (L_n)}
      &= bab^2\cdot ab^nw_n^* + b^2\cdot ab^nw_n^* + \dots + b^nw_n^*
    \end{align*}
    Suppose there is a \textit{suffixes}-enhanced regular expression
    $r$ for $\LQ{a}{L}$ that only depends on $a$ and $\Sigma$ and that refers to
    finitely many derivatives, say, $\LQ{v_1}{L}, \dots,
    \LQ{v_m}{L}$. Considering $r$ for $L_n'$, we find that $r$ cannot contain the \textit{suffixes}
    function because that would introduce words starting with $a$,
    which cannot be in $L_n'$ and which cannot be amended by
    prepending a fixed word without breaking the $a$-$b$ pattern.
    There must exist some  $v \in w_n^*$ such that
    each $v_j$ is either a prefix of $v$ that ends with an $a$ or it
    is not a prefix of $v$.

    Now, if we consider $L_{k}$ where $k=\max (n, |v_1|, \dots,
    |v_m|)+1$ then none of the $\LQ{v_j}{L_k}$ can contain
    $b^kw_k^*$. Note that if $v_j$ is not a prefix of $w_n*$, then it
    is not a prefix of $w_k^*$, for any $k\ge n$, either.
    Hence, $r$ cannot describe $L_k'$.
    
    If we assume that $\calF$ contains
    \textit{suffixes} and the left quotient operator, then we could consider $\textit{suffixes}(L)$
    as an abbreviation for $\LQ{\Sigma^*}{L}$ and we would regain
    left derivability. Furthermore, with a suitable variation of
    Definition~\ref{def:left-closed}, \textit{suffixes} is \emph{right derivable}:\\ 
    $\RQ{\textit{suffixes} (L)}{a}  = \RQ{\{v \mid \exists u. uv \in L\}}{a} = \{v \mid \exists u. uva
    \in L\} = \textit{suffixes} (\RQ{L} a)$.\\
    The function \textit{prefixes} is left derivable: \\
    $\LQ{a}{\textit{prefixes} (L)} = \LQ{a}{\{v \mid \exists u. vu \in L\}} = \textit{prefixes}
    (\LQ{a}{L})  = \hcalJ (\textit{prefixes} (x_{a,1}))$.
  \item The function \textit{reverse} is neither left derivable nor right derivable, but swaps between left and right
    quotients:\\
    $\LQ{a}{\textit{reverse} (L)} = \textit{reverse} (\RQ{L}{a})$.\\
    To see that \textit{reverse} is no left derivable, consider the
    language $L= b^*a$. Clearly, $\mathit{reverse} (L) = ab^*$ and
    $\LQ{a}{\mathit{reverse} (L)} = b^*$. Now suppose we can obtain
    $b^*$ by a regular expression with \textit{reverse} on arbitrary
    derivatives of $L$. There are only two distinct derivatives: $\LQ{a}{(b^*a)}
    = \{\varepsilon\}$ and $\LQ{b}{(b^*a)} = {b^*a}$. Hence, for any
    $w\in\{a,b\}^*$, $\LQ{w}{(b^*a)}$ will be either empty,
    $\{\varepsilon\}$, or $b^*a$. Now consider a language $U$
    constructed from these derivatives by application of regular
    operators or \textit{reverse}. It can be shown that any word in
    $U$ is either $\varepsilon$ or it contains the symbol $a$. Thus,
    $U$ cannot be equal to $b^*$.
  \item The enhancement with the approximation operators $\mathbb{H}_k, \mathbb{H}_{k-1}, \dots,
    \mathbb{H}_1, \mathbb{H}_0$ operators 
    (for Hamming distance) is left derivable because
    \begin{displaymath}
      \LQ{a}{\mathbb{H}_k (L)} = \mathbb{H}_k (\LQ{a}{L}) + \sum_{\substack{k>0 \\ x\ne a}} \mathbb{H}_{k-1} (\LQ{x}{L})
    \end{displaymath}
    For approximation with operators $\mathbb{L}_k, \dots, \mathbb{L}_0$ that rely on Levenshtein
    distance, we also obtain left closure (assuming that
    $\mathbb{L}_{-1} (L) = \emptyset$):
    \begin{displaymath}
      \LQ{a}{\mathbb{L}_k (L)} =
      \sum_{
        \substack{
          w \in \Sigma^*\\ 
          |w| \le k\\
          k>0
        }
      } \Big(
      \mathbb{L}_{k-|w|} (\LQ{wa}{L}) +
      \sum_{x\ne a} \mathbb{L}_{k-|w|-1} (\LQ{wx}{L}) +
      \mathbb{L}_{k-|w|-1} (\LQ{w}{L})
      \Big)
    \end{displaymath}
    The terms correspond to the actions ``delete $w$, then match $a$'', ``delete $w$, then
    substitute $a$ by some $x$'', and ``delete $w$, then insert $a$''.
  \item Tilde and bar are trivially left derivable: $\LQ{a} \tilde{L}
    = \LQ{a}L$ and $\LQ{a} \bar{L} = \LQ{a} L$.
  \end{enumerate}
\end{example}

\section{Word Problem}
\label{sec:semantic-characterization}

To obtain a decision procedure for the word problem of left derivable enhanced regular expressions,
we first define the corresponding syntactic derivative and then extend Brzozowski's result that $w
\in \L (r)$ iff $\varepsilon \in \L (D (w, r)) $ (which follows from
Theorem~\ref{theorem:word-derivative}). It is interesting to remark that, for example, we obtain a
decision procedure for the word problem for the language of regular expressions enhanced with the
shuffle-closure operator is no longer regular. 
\begin{theorem}\label{thm:left-closed-derivative}
  If $(\Power(\Sigma^*), \calJ)$ is a left derivable $\calF$-algebra which is $\varepsilon$-testable, then
  there is a syntactic derivative function
  $D : \Sigma \times T_{\calF} \to T_{\calF}$ such that
  $\hcalJ (D (a, t)) = \LQ{a}{\hcalJ (t)}$, for all $a\in\Sigma$ and
  $t\in T_\calF$.
\end{theorem}
\begin{proof}
  Define $D$ inductively as an extension of
  Definition~\ref{definition:brzozowski-derivative} for $F\in \calF\setminus\calR$: 
  \begin{align*}
    D (a,F (r_1, \dots, r_n)) &= R (F, a)[x_{v,j} \mapsto D (v, r_j) \mid x_{v,j} \in X (F,a) ]
  \end{align*}
  where $N$ extends to $T_{\calF}$  by Lemma~\ref{lemma:extended-nullability} and where $D$ extends
  to words as before.
\begin{align*}
    D (\varepsilon, r) &= r & D (aw, r) &= D (w, D (a, r))
  \end{align*}
  The statement about the semantics follows by induction on the augmented term using the definition
  of left derivability.
  \qed
\end{proof}

\begin{theorem}
  If $(\Power(\Sigma^*), \calJ)$ is a left derivable $\calF$-algebra which is $\varepsilon$-testable, then
  the word problem for $\hcalJ (t)$ is decidable, for any $t \in T_{\calF}$.
\end{theorem}
\begin{proof}
  By Theorem~\ref{thm:left-closed-derivative}, there is a nullability function $N$ and a derivative
  $D$ for $T_{\calF}$. By induction on the length of $w \in \Sigma^*$, we obtain that $w
  \in \hcalJ (t)$ iff $\varepsilon \in \hcalJ (D (w, t))$ iff $N (D (w,t))$.
  \qed
\end{proof}

\section{Finiteness}

For classical derivatives on $T_\calR$ (cf.\
Definition~\ref{definition:brzozowski-derivative}), Brzozowski showed
that the set of iterated derivatives $D (\Sigma^*, r)$ of a given
regular expression $r$ is finite, 
when considered modulo similarity (i.e., associativity, commutativity, and idempotence of union). Hence, we now look
for conditions such that the set of dissimilar iterated derivatives is
finite for enhanced regular expressions. First, we  set up a
framework for reasoning about finiteness.

\begin{figure}[tp]
  \begin{displaymath}
    \begin{array}[t]{ll}
    D^+ (\Rnull) &= \{ \Rnull \} \\
    D^+ (\Rempty) &= \{ \Rnull \} \\
    D^+ (a) &= \{ \Rnull, \Rempty \} 
    \end{array}
    \qquad
    \begin{array}[t]{ll}
    D^+ (r+s) &= D^+ (r) \oplus D^+ (s) \\
    D^+ (r\cdot s) &= D^+ (r) \odot s \oplus \bigoplus D^+ (s) \\
    D^+ (r^*) &= \bigoplus D^+ (r) \odot r^*
    \end{array}
  \end{displaymath}
  \caption{Iterated Brzozowski derivatives for $T_\calR$}
  \label{fig:iterated-derivatives}
\end{figure}
Recent work on determining the number of iterated \emph{partial} derivatives starts with an inductive
definition for the set of iterated partial derivatives
\cite{broda:tr}. We transfer that definition to the classical case and
define an upper approximation $D^+ (r)$ of the set of iterated
derivatives of expression $r$ in
Figure~\ref{fig:iterated-derivatives} by induction on  $r$. In the definition, we lift $\odot$ and $\oplus$ to
sets of expressions (i.e., if $R, S \subseteq T_\calR$, then $R\odot S
= \{ r\odot s \mid r\in R, s\in S\}$ and $R\oplus S = \{ r \oplus s
\mid r\in R, s\in S$). We
further write $\bigoplus S$ for the set $\{ s_1 \oplus \dots \oplus s_n \mid n \in \nat, s_i \in S
\}$ of finite sums of elements drawn from $S$ where the nullary sum
stands for $\Rnull$ and where we assume sums to be identified modulo
associativity, commutativity, and idempotence to obtain the following results.\footnote{See the
  technical report for auxiliary lemmas and proofs.}
\begin{theorem}
  The set $D^+ (r)$ is finite, for all $r \in T_\calR$.
\end{theorem}
Clearly, the set $D^* (r)  = \{r\} \cup D^+ (r)$ is also finite for all $r$.

\begin{theorem}[Closure under derivation]\label{theorem:closure-under-derivation}
  \begin{enumerate}
  \item For all $r$ and $a$, $D (a, r) \in D^+ (r)$.
  \item For all $r$ and $a$, if $t \in D^+ (r)$, then $D (a, t) \in D^+ (r)$.
  \end{enumerate}
\end{theorem}

\begin{corollary}
  The set $\{ D (w, r) \mid w\in \Sigma^+ \} \subseteq D^+ (r)$, for
  all $r$.
\end{corollary}

To obtain finiteness for enhanced regular expressions, we strengthen the notion of left
derivability. Essentially, we restrict the form of a derivative to a
linear combination of enhancement functions applied to derivatives of
the arguments. 

\begin{definition}
  Let $\calF = \{ F_1, \dots, F_m \}$ be a ranked alphabet.  The $\calF$-algebra
  $(\Power (\Sigma^*), \calJ)$ is \emph{linear left derivable} if, for each $n\in\nat$,
  $F\in \calF^{(n)}$, and $a\in \Sigma$, there exists a finite index set $J$ such that, for each
  $j\in J$, there is a word $v_j \in \Sigma^*$, an index $i_j \in \{1,\dots, m \}$ of an element of
  $\calF$ with arity $\Arity (F_{i_j}) = n_j$, and, for $1\le k \le n_j$, words $w^j_k \in \Sigma^*$
  and indexes $\alpha^j_k \in \{ 1, \dots, n\}$ of left-hand-side languages,
  such that for all $L_1, \dots L_n \subseteq \Sigma^*$, the left quotient can be expressed by:
  \begin{align}
    \LQ{a} (\calJ (F) (L_1, \dots, L_n)) &=
    \bigcup_{j\in J} v_j\cdot \calJ(F_{i_j}) (\overline{\LQ{w^j_{k}}( L_{\alpha^j_k}))}^{k=1,\dots,n_j})
  \end{align}
\end{definition}
Of the standard regular operators, only union (and in fact all boolean functions) is linear left
derivable. Concatenation $U \cdot V$ does not fit the pattern because it has a summand which is
conditional on $\varepsilon \in U$. The Kleene star does not fit, either, because it concatenates
the derivative of the argument with the  original term
(Definition~\ref{definition:brzozowski-derivative}). But many useful operators are linear left 
derivable (Example~\ref{example:linear-left-derivable}). 

\begin{theorem}\label{theorem:finite-derivative}
  Suppose that $\calF = \{ F_1, \dots, F_m \} \cup \calR$ is an enhanced
  regular alphabet with interpretation $\calJ$ such that $(\Power
  (\Sigma^*), \calJ_{|\{ F_1, \dots, F_m \}})$ is linear left derivable.

  Then, for all $n\in\nat$, $F \in \calF^{(n)}$, and $a\in\Sigma$
  there exists a finite index set $J$, for each $j\in J$, there is a word $v_j \in \Sigma^*$, an
  index $i_j \in \{1,\dots, m \}$
  of an element of $\calF\setminus \calR$ with arity $\Arity (F_{i_j}) = n_j$,
  for each $1 \le k \le n_j$, a word $w^j_k  \in \Sigma^*$, and an index $\alpha^j_k \in \{ 1,
  \dots, n\}$ that   selects one of the 
  left-hand-side regular expressions as an argument. 
  Then, for each $r_1, \dots, r_n \in T_{\calF}$, the syntactic derivative of $F (r_1, \dots, r_n)$
  by $a$ is given in the form
  \begin{align}
    \label{eq:2}
    D (a, F (r_1, \dots, r_n)) &=
    \sum_{j\in J(F,a)} v_j\cdot F_{i_j} (\overline{D (w^j_{k}, r_{\alpha^j_k})}^{k=1,\dots,n_j})
  \end{align}

  In this setting, the set of iterated derivatives of any $\calF$-regular expression $r$ is
  finite. Specifically, in extension of the definition in Figure~\ref{fig:iterated-derivatives}, we claim that for
  each $F\in \calF\setminus \calR$, the set of iterated derivatives
  \begin{align}
    \label{eq:1}
    D^+ (F (r_1,\dots, r_n))
    &= 
    \bigoplus \{ v\cdot G  (\overline{r_i'}) \mid V \IsExtensionOrEqual v, G \in \calF\setminus\calR, r_i' \in \bigcup_j D^* (r_j)  \}
  \end{align}
  is finite.   
  Here $V = \{ v_j \mid j\in J, F \in \calF, a\in \Sigma \}$, and we write $V
  \IsExtensionOrEqual v$ for $\exists v'\in V. v'\IsExtensionOrEqual v$.
\end{theorem}

\begin{corollary}\label{corollary:regular}
  Let $\calF$ be an enhanced regular alphabet and $(\Power (\Sigma^*),
  \calJ)$ be an $\varepsilon$-testable, linear left derivable
  $\calF$-algebra. Then any $\calF$-regular expression defines a
  regular language. 
\end{corollary}
\begin{proof}
  Let $r\in T_\calF$ and let $Q_r$ be the set of dissimilar derivatives of $r$. As $Q_r \subseteq
  D^* (r)$, $Q_r$ is finite. Hence $M = (Q_r, \Sigma, D, r, F)$ with $F = \{ q \in Q_r \mid N (q) \}$ is
  a total  deterministic finite automaton that recognizes $\L (r)$,
  which is thus regular.
  \qed
\end{proof}
\begin{corollary}
  Let $\calF$ be an enhanced regular alphabet and $(\Power (\Sigma^*),
  \calJ)$ be an $\varepsilon$-testable, linear left derivable
  $\calF$-algebra. Then, for each $F \in \calF$, the operation $\calJ (F)$ preserves
  regularity. 
\end{corollary}
\begin{proof}
  Let $F \in \calF^{(n)}$, for some $n \in \nat$.
  Let $R_1, \dots, R_n$ be regular languages defined by regular expressions $r_1, \dots, r_n \in
  T_\calR \subseteq T_\calF$.
  By Corollary~\ref{corollary:regular}, $\hcalJ (F (r_1,\dots, r_n))$ is regular.
  Hence $\calJ (F)$ preserves regularity.
  \qed
\end{proof}

\begin{example}\label{example:linear-left-derivable}
  Many operators are in fact linear left derivable.
  \begin{enumerate}
  \item Intersection and complement are linear left derivable.
  \item The shuffle operation is linear left derivable, but the
    derivative of the shuffle closure
    contains a nested application of shuffle closure.
  \item Inverse and non-erasing homomorphism are linear
    left derivable.
  \item For $k>0$, the set $\{ f_{i,k} \mid 0<i\le k\}$ is linear left derivable.
  \item The left quotient is not linear left derivable, but the right quotient is linear left derivable.
  \item The function \textit{suffixes} is not left derivable; the function \textit{prefixes} is
    linear left derivable.
  \item The function \textit{reverse} is not left derivable.
  \item Both, $\mathbb{H}_k$ and $\mathbb{L}_k$ are linear left derivable.
  \item Tilde and bar are linear left derivable.
  \end{enumerate}
  By Corollary~\ref{corollary:regular}, regular languages are
  closed under $\varepsilon$-testable operators that are linear
  left derivable:
  $\cap$, $\neg$, $\|$, $h^{-1}$, non-erasing $h$, $\mathbb{H}_k$,
  tilde, bar.
\end{example}

For a set of unary operators, linear left derivability amounts to definability
by a rational finite state transducer.
\begin{theorem}\label{thm:transducer-from-linear-left-derivability}
  Let $\calF = \calF^{(1)} = \{F_1, \dots, F_m\}$ be a ranked alphabet
  of unary operators and $(\Power (\Sigma^*), \calJ)$ be a linear left
  derivable $\calF$-algebra which is $\varepsilon$-testable using the
  identity function.
  Then, for each $1\le l\le m$ and $L\subseteq\Sigma^*$, $\calJ (F_l)
  (L)$ is equal to $T (L)$ where $T$ is a rational finite state transducer.
\end{theorem}

The reverse implication does not hold because transducers may, in general, consume
an unbounded amount of input before producing an output. The
transducers resulting from
Theorem~\ref{thm:transducer-from-linear-left-derivability} only
consume bounded input before producing at least one output symbol.

\section{Conclusion}

We introduce a framework for constructing derivatives for
regular expressions enhanced with new operators. If these operators are left
derivable, we obtain an algorithm for the word problem; if they are
linear left derivable, we can construct a DFA from an enhanced
expression. In fact, unary operators with this property are rational transductions.

Some of the operators considered in this paper are known to be regularity preserving, yet, they
fail to be linear left derivable or to be $\varepsilon$-testable.  In
future work, we plan to address these restrictions by generalizing linear derivability as well as the nullability test.

\bibliography{main}

\begin{thebibliography}{10}
\providecommand{\url}[1]{\texttt{#1}}
\providecommand{\urlprefix}{URL }

\bibitem{AntimirovIneq}
Antimirov, V.M.: Rewriting regular inequalities. In: Proc.\ of FCT'95. LNCS,
  vol. 965, pp. 116--125. Springer (1995)

\bibitem{Antimirov96Partial}
Antimirov, V.M.: Partial derivatives of regular expressions and finite
  automaton constructions. Theoretical Computer Science  155(2),  291--319
  (1996)

\bibitem{broda:tr}
Broda, S., Machiavelo, A., Moreira, N., Reis, R.: Study of the average size of
  {Glushkov} and partial derivative automata (Oct 2011)

\bibitem{321249}
Brzozowski, J.A.: Derivatives of regular expressions. J. ACM  11(4),  481--494
  (1964)

\bibitem{CaronChamparnaudMignot2011}
Caron, P., Champarnaud, J.M., Mignot, L.: Partial derivatives of an extended
  regular expression. In: LATA. LNCS, vol. 6638, pp. 179--191. Springer (2011)

\bibitem{DBLP:conf/wia/CaronCM12}
Caron, P., Champarnaud, J., Mignot, L.: Multi-tilde-bar derivatives. In:
  {CIAA}. LNCS, vol. 7381, pp. 321--328. Springer (2012)

\bibitem{DBLP:journals/ita/CaronCM14}
Caron, P., Champarnaud, J., Mignot, L.: A general framework for the derivation
  of regular expressions. {RAIRO} - Theor. Inf. and Applic.  48(3),  281--305
  (2014)

\bibitem{DBLP:conf/lata/ChamparnaudJM12}
Champarnaud, J., Jeanne, H., Mignot, L.: Approximate regular expressions and
  their derivatives. In: {LATA}. LNCS, vol. 7183, pp. 179--191. Springer (2012)

\bibitem{ChamparnaudMignotNicart2015}
Champarnaud, J., Mignot, L., Nicart, F.: Constrained expressions and their
  derivatives. Tech. Rep. 1406.6144v2, arXiv (Oct 2015),
  \url{http://arxiv.org/abs/1406.6144v2}

\bibitem{Grabmayer:2005:UPC:2156157.2156171}
Grabmayer, C.: Using proofs by coinduction to find "traditional" proofs. In:
  Proc.\ of CALCO'05. pp. 175--193. Springer (2005)

\bibitem{DBLP:conf/icfp/MightDS11}
Might, M., Darais, D., Spiewak, D.: Parsing with derivatives: a functional
  pearl. In: Proc.\ of ICFP'11. pp. 189--195. ACM (2011)

\bibitem{re-derivs}
Owens, S., Reppy, J., Turon, A.: Regular-expression derivatives reexamined.
  Journal of Functional Programming  19(2),  173--190 (2009)

\bibitem{DBLP:conf/rta/RouV03}
Rosu, G., Viswanathan, M.: Testing extended regular language membership
  incrementally by rewriting. In: RTA'03. LNCS, vol. 2706, pp. 499--514.
  Springer (2003)

\bibitem{DBLP:conf/lata/SulzmannT15}
Sulzmann, M., Thiemann, P.: Derivatives for regular shuffle expressions. In:
  LATA'15. LCNS, vol. 8977, pp. 275--286. Springer, Nice, France (Mar 2015)

\bibitem{SulzmannThiemann2016-lata}
Sulzmann, M., Thiemann, P.: Forkable regular expressions. In: Proc.\ of
  LATA'16. Prague, Czech Republic (2016)

\bibitem{DBLP:journals/jfp/TraytelN15}
Traytel, D., Nipkow, T.: Verified decision procedures for {MSO} on words based
  on derivatives of regular expressions. J. Funct. Program.  25 (2015)

\bibitem{DBLP:conf/wia/Watson96}
Watson, B.W.: {FIRE} lite: {FAs} and {REs} in {C++}. In: Workshop on
  Implementing Automata. LNCS, vol. 1260, pp. 167--188. Springer (1996)

\end{thebibliography}

\appendix{}
\section{Proofs and auxiliary lemmas}

\begin{proof}[of Lemma~\ref{lemma:extended-nullability}]
  If $f = \calJ (F)$ for $F\in\calF^{(n)}$ is $\varepsilon$-testable, then there exists a boolean
  function $B_f$ such that $\varepsilon\in f (L_1, \dots, L_n)$ iff $B_f (\varepsilon \in L_1,
  \dots, \varepsilon \in L_n)$. Now set
  \begin{displaymath}
    N (F (t_1, \dots, t_n)) := B_f (N (t_1), \dots, N (t_n))
  \end{displaymath}
  for the desired extension of $N$.   \qed 
\end{proof}

\begin{proof}[of Theorem~\ref{theorem:finite-derivative}]
  In equation~\eqref{eq:1}, the set $V$ is finite as it is a union of finite sets.
  Hence, the set $\{ v \mid V \IsExtensionOrEqual v \}$ is also finite.
  As $\calF$ and $D^+ (r_j)$ are also assumed finite, the
  set $D^+ (F (r_1,\dots,r_n))$ is finite.

  The set $D^+ (F (r_1,\dots,r_n))$ is also closed under formation of derivatives. Consider the
  derivative of a summand of the form $D(a,v\cdot F (r'_1,\dots,r'_n))$:
  \begin{align*}
    D(a,v\cdot F (\overline{r'_i})) &=
    \begin{cases}
      \Rnull & v = bv', a\ne b \\
      v'\cdot F (\overline{r'_i}) & v=av' \\
      \sum_{j\in J(F,a)} v_j\cdot F_{i_j} (\overline{D (w^j_{k},
        r'_{\alpha^j_k})}^{k=1,\dots,n_j}) & v=\varepsilon
    \end{cases}
  \end{align*}
  In each case, the outcome is covered by the right hand side of Equation~\eqref{eq:1}.
  \qed
\end{proof}

We need a few auxiliary lemmas before we can prove that $D^+ (r)$ is
closed under the derivative operation.
\begin{lemma}\label{lemma:null-in-dplus}
  For all $r\in T_\calR$, $\Rnull \in D^+ (r)$.
\end{lemma}
\begin{proof}[Lemma~\ref{lemma:null-in-dplus}]
  Induction on $r$.

  \textbf{Case }$\Rnull$, $\Rempty$, $a$: Immediate.

  \textbf{Case }$r+s$: $\Rnull = \Rnull \oplus \Rnull \in D^+ (r) \oplus D^+ (s)$, by induction.

  \textbf{Case }$r\cdot s$: $\Rnull = \Rnull \odot s \oplus \Rnull \in D^+ (r)\odot s \oplus D^+ (s)$, by
  induction.

  \textbf{Case }$r^*$: $\Rnull = \Rnull \odot r^* \in D^+ (r) \odot r^*$.
\end{proof}

\begin{lemma}\label{lemma:dplus=plusd}
  For all $r$ and $s$, $D (a, r \oplus s) = D (a, r) \oplus D (a,s)$.
\end{lemma}
\begin{proof}[Lemma~\ref{lemma:dplus=plusd}]
  By cases in the definition of $\oplus$.

  \textbf{Case }$s=\Rnull$. $D (a, r \oplus \Rnull) = D (a, r) = D (a, r) \oplus \Rnull = D (a, r)
  \oplus D (a, \Rnull)$.

  \textbf{Case }$r=\Rnull$. Analogously.

  \textbf{Case }$r=s$. $D (a, r \oplus r) = D (a, r) = D(a, r) \oplus D (a,r)$.

  \textbf{Case }$r\oplus s = r+s$. $D (a, r \oplus s) = D (a, r+s) = D (a, r) \oplus D (a, s)$.
\end{proof}
\begin{lemma}\label{lemma:dodot=ddot}
  For all $r$ and $s$,  $D (a, r\odot s) = D (a, r\cdot s)$.
\end{lemma}

\begin{proof}[Lemma~\ref{lemma:dodot=ddot}]
  By cases in the definition of $\odot$.

  \textbf{Case }$r = \Rnull$.
  $D (a, r\odot s) = D (a, \Rnull) = \Rnull$.
  $D (a, r\cdot s) = D (a, \Rnull \cdot s) = D
  (a, \Rnull)\odot s \oplus N (\Rnull) \odot D (a, s) = \Rnull \oplus \Rnull = \Rnull$.

  \textbf{Case }$s= \Rnull$. Similar.

  \textbf{Case }$r = \Rempty$.
  $D (a, r\odot s) = D (a, s)$.
  $D (a, r \cdot s) = D (a, \Rempty \cdot s) = D (a, \Rempty) \odot s \oplus N (\Rempty) \odot D (a, s) = \Rnull
  \oplus D (a,s) = D (a, s)$.

  \textbf{Case }$s = \Rempty$. Similar.

  \textbf{Case }$r\odot s = r\cdot s$. Immediate.
\end{proof}

\begin{proof}[of Theorem~\ref{theorem:closure-under-derivation}]
  \textbf{Part 1.} By induction on $r$.

  \textbf{Case }$\Rnull$: $D (a, \Rnull) = \Rnull \in D^+ (\Rnull)$.

  \textbf{Case }$\Rempty$: $D (a, \Rempty) = \Rnull \in D^+ (\Rempty)$.

  \textbf{Case }$b$: $D (a, b) \in \{\Rnull, \Rempty\} = D^+ (b)$.

  \textbf{Case }$r+s$: immediate.

  \textbf{Case }$r\cdot s$: $D (a, r\cdot s) = D (a, r) \odot s \oplus N (r)\odot D (a,s)$.
  By induction, $D (a,r) \in D^+ (r)$ and $D (a,s) \in D^+ (s)$. We distinguish two cases on the
  outcome of $N (r)$.

  \textbf{Subcase }$N (r) = \Rnull$: In this case $D (a, r\cdot s) = D (a, r) \odot s \oplus \Rnull \in D^+ (r) \odot
  s \oplus \Rnull$, by induction, and by Lemma~\ref{lemma:null-in-dplus}, $\Rnull\in D^+ (s)$, so $D^+ (r) \odot
  s \oplus \Rnull  \subseteq D^+ (r) \odot  s \oplus D^+ (s)$.

  \textbf{Subcase }$N (r) = \Rempty$: In this case $D (a, r\cdot s) = D (a, r) \odot s \oplus D (a,s) \in D^+ (r) \odot
  s \oplus D^+ (s)$, by induction.

  \textbf{Case }$r^*$: $D (a, r^*) = D (a, r) \odot r^* \in D^+ (r) \odot r^*$ by induction.

  \textbf{Part 2.} By induction on $r$. 

  \textbf{Case }$\Rnull$, $\Rempty$, $a$: immediate.

  \textbf{Case }$r+s$: if $t \in D^+ (r+s)$, then $t = r' \oplus s'$ for some $r' \in D^+ (r)$ and
  $s' \in D^+ (s)$. By Lemma~\ref{lemma:dplus=plusd}, $D (a, r' \oplus s') = D (a, r') \oplus D (a, s') \in D^+
  (r) \oplus D^+ (s)$, by induction.

  \textbf{Case }$r\cdot s$: if $t \in D^+ (r\cdot s)$, then $t = r' \odot s \oplus s'$ for some $r' \in D^+
  (r)$ and $s' \in D^+ (s)$. By Lemmas~\ref{lemma:dplus=plusd} and~\ref{lemma:dodot=ddot}, $D (a, r'
  \odot s \oplus s') = D (a, r' \odot s) \oplus D (a, s') = D (a, r'\cdot s) \oplus D (a, s')
  = D (a, r') \odot s \oplus N (r') \odot D (a,s) \oplus D (a,s')$.

  By induction, $D (a, r') \odot s \in D^+ (r) \odot s$.

  By item~1, $D (a,s) \in D^+ (s)$.

  By induction, $D (a, s') \in D^+ (s)$.

  Hence, $N (r') \odot D (a,s) \oplus D (a,s') \in \bigoplus D^+ (s)$, which proves the claim.

  \textbf{Case }$r^*$: if $t \in D^+ (r^*)$, then $t = r' \odot r^*$ for some $r' \in D^+ (r)$. Now,
  $D (a, r' \odot r^*) = D (a, r'\cdot r^*) = D (a, r') \odot r^* \oplus N (r') \odot D (a, r) \odot r^*$.
  
  By induction, $D (a, r') \odot r^* \in D^+ (r) \odot r^*$.

  By item~1, $D (a,r) \odot r^* \in D^+ (r) \odot r^*$.

  Hence, $D (a, r') \odot r^* \oplus N (r') \odot D (a, r) \odot r^* \in \bigoplus D^+ (r) \odot r^*$.
\end{proof}

\clearpage{}
\begin{proof}[Theorem~\ref{thm:transducer-from-linear-left-derivability}]
  For $F_l$, define the finite state transducer $T= (Q, \Sigma,
  \Sigma, I, A, \delta)$ as follows: 
  set of states $Q = \{F_1, \dots, F_m\}$; initial states $I = \{ F_l
  \}$;
  accepting states $A = Q$;
  the transition relation is the smallest relation $\delta\subseteq Q\times \Sigma^* \times
  \Sigma^* \times Q$ (relating a state and an input word to an output
  word and a next state) such that: For $F\in\calF$ and $a\in \Sigma$,
  if, by linear left derivability,  $$
  \LQ{a} (\calJ (F) (L)) = \bigcup_{j\in J
    (F, a)} v_j \cdot \calJ (F_{i_j}) (\LQ{w_j}L)
  $$ where $v_j, w_j
  \in \Sigma^*$, then
  $
  \delta \supseteq \{ (F, w_j, a\cdot v_j, F_{i_j}) \mid j\in J
  (F, a) \}
  $.

  Now recall the definition of $T_l (L)$ (where $l=1,\dots,m$), the
  language translated from state $F_l$, for the transducer  $T= (Q, \Sigma,
  \Sigma, I, F, \delta)$: 
  \begin{displaymath}
    v \in T_i (L) \Leftrightarrow
    \exists w \in L: \exists k\in\{1,\dots, m\}:
    (F_i, v, w, F_k) \in \delta^*
  \end{displaymath}
  where $\delta^* \subseteq Q \times \Sigma^* \times \Sigma^* \times
  Q$ is the iterated transition relation defined as the smallest
  relation such that:
  \begin{enumerate}
  \item\label{item:1} for all $F\in Q$, $(F, \varepsilon, \varepsilon, F) \in
    \delta^*$;
  \item\label{item:2} for all $F_i, F_j, F_k \in Q$, $v', v  , w', w\in \Sigma$, \\
    $(F_i, v'v, w'w, F_k) \in
    \delta^*$ iff $(F_i, v', w', F_j) \in \delta$ and $(F_j, v, w,
    F_k) \in \delta^*$.
  \end{enumerate}
  
  Now suppose that $v \in \calJ (F) (L)$ and show that $v\in T (L)$ by
  induction on the length of $v$.

  If $v=\varepsilon$, then $\varepsilon$-testability with the identity
  function implies that $\varepsilon\in L$. Hence $\varepsilon \in T
  (L)$ by the definition of $\delta^*$ (Item~\ref{item:1}).

  If $v\ne\varepsilon$, then consider $\LQ{a}v \in \LQ{a} (\calJ (F)
  (L))$. By assumption of linear left derivability, we find that 
  $$
  \LQ{a} (\calJ (F) (L)) = \bigcup_{j\in J
    (F, a)} v_j \cdot \calJ (F_{i_j}) (\LQ{w_j}L)
  $$ for some $v_j, w_j\in \Sigma^*$. Thus, there exists some $j\in J
  (F, a)$ such that 
  \begin{align*}
    \LQ{a}v &\in v_j \cdot \calJ (F_{i_j}) (\LQ{w_j}L)
  \end{align*}
  Hence, there is some $v'\in\Sigma^*$ such that $v = a\cdot v_j\cdot v'$ with
  $v'  \in \calJ (F_{i_j}) (\LQ{w_j}L)$. By induction, $v' \in T
  (\LQ{w_j}L)$ which means that there exists some $w' \in \LQ{w_j}L$
  and $F_k\in Q$ 
  such that $(F_{i_j}, w' , v', F_k) \in \delta^*$. By construction,
  $(F, w_j, a\cdot v_j, F_{i_j}) \in \delta$, so by definition of
  $\delta^*$ (Item~\ref{item:2}), $(F, w_j \cdot w',
  a\cdot v_j\cdot v', F_k) \in \delta^*$ and thus $v\in
  T (L)$.

  The reasoning for the reverse inclusion $T (L) \subseteq \calJ (F)
  (L)$ is analogous.
  \qed
\end{proof}

\clearpage
\section{On upward and downward closures}
\label{sec:upward-downw-clos}

The upward and downward closure of any language is regular \verb.\cite{HAINES196994}.: $\Upward{L} = \{ y
    \mid \exists x \in L. x \sqsubseteq y \}$ and $\Downward{L} \{ x
    \mid \exists y \in L. x \sqsubseteq y \}$. Here, ${\sqsubseteq}
    \subseteq \Sigma^*\times \Sigma^*$ is
    the \emph{subword ordering} where
    $x \sqsubseteq y$ iff $x =
    a_1\cdots a_n$ and $y = u_0a_1u_1\cdots u_{n-1}a_nu_n$ for some
    $n\in\nat$, $a_i\in\Sigma$, for $1\le i \le n$, and
    $u_i\in\Sigma^*$, for $0 \le i\le n$.

The upward closure is $\varepsilon$-testable using the
    identity function, but the downward closure is not, in fact,
    $\varepsilon \in \Downward{L}$ iff $L\ne \emptyset$.
    To see this, consider $L_1 = \emptyset$ and $L_2 = \{a\}$ with
    downward closures $\Downward{L_1} = \emptyset$ and $\Downward{L_2}
    = \{\varepsilon, a\}$. Thus, a hypothetical boolean function
    $B_\downarrow$ would have to satisfy $B_\downarrow (\False) =
    \False$ (according to $L_1$) and  $B_\downarrow (\False) =
    \True$ (according to $L_2$), a contradiction.

The upward closure operation is left derivable as $\LQ{a} (\Upward{L}) =
    \Upward{L} \cup \Upward{(\LQ{a}L)}$.\\
    The downward closure operation does not appear to be left derivable as the
    left quotient involves skipping ahead arbitrarily many characters:
    $\LQ{a} (\Downward{L}) = \sum_{w \in \Sigma^*}
    \Downward{(\LQ{wa}L)}$, but we have neither proof nor disproof for this conjecture.

Upward closure is linear left derivable, but downward closure
    is not.

\end{document}